\documentclass[a4paper,oneside,12pt]{article}

\usepackage{jheppub}
\usepackage{amsfonts}
\usepackage{amsmath}
\usepackage{amssymb}
\usepackage{amsthm}
\usepackage[mathscr]{eucal}
\usepackage{bbold}
\usepackage{braket}
\usepackage{color}
\usepackage{dsfont}
\usepackage{framed}
\usepackage{mathtools}
\usepackage{physics}
\usepackage[normalem]{ulem}
\usepackage{tensor}
\usepackage{thmtools}
\usepackage{thm-restate}
\usepackage{ulem}
\usepackage{tikz}
\usetikzlibrary{arrows.meta,arrows} 
\usepackage{centernot}
\usepackage{enumitem} 
\usepackage[a]{esvect}  
\usepackage{slashed}
\usepackage[only,llbracket,rrbracket]{stmaryrd} 

\usepackage{yfonts}
\usepackage{float}
\usetikzlibrary{math} 
\usetikzlibrary{calc}		

\DeclareFontFamily{U}{jkpmia}{}
\DeclareFontShape{U}{jkpmia}{m}{it}{<->s*jkpmia}{}
\DeclareFontShape{U}{jkpmia}{bx}{it}{<->s*jkpbmia}{}
\DeclareMathAlphabet{\mathfrak}{U}{jkpmia}{m}{it}
\SetMathAlphabet{\mathfrak}{bold}{U}{jkpmia}{bx}{it}

\usepackage{subcaption}
\usepackage{float}
\usepackage{afterpage}
\renewcommand{\thefigure}{\arabic{figure}}

\usetikzlibrary{arrows}
\tikzset{elabelcolor/.style={color=blue} 
    vertex/.style={circle,draw,minimum size=1.5em},
    edge/.style={->,> = latex'}
}  

\usepackage{cleveref}

\usepackage{longtable}

\usepackage[ruled,vlined]{algorithm2e}

\definecolor{VHcolor}{rgb}{0.7,0.3,0.9}

\definecolor{MRocolor}{rgb}{0.1,0,1}

\newtheorem{defi}{Definition}
\newtheorem{lemma}{Lemma}

\newtheorem{thm}{Theorem}

\newcommand{\N}{{\sf N}}
\newcommand{\D}{\sf{D}}

\newcommand{\X}{X}
\newcommand{\Y}{Y}
\newcommand{\Z}{Z}
\newcommand{\W}{W}

\newcommand{\nsp}{\llbracket\N\rrbracket} 

\newcommand{\mgs}{\mathcal{E}} 
\newcommand{\ds}{\mathcal{D}} 
\newcommand{\cds}{\mathcal{D}^\complement} 
\newcommand{\ac}{\Check{\mathcal{A}}}  

\newcommand{\bset}{\beta\text{-set}} 
\newcommand{\bsets}{\beta\text{-sets}}
\newcommand{\bs}[1]{\beta(#1)}
\newcommand{\ess}{\mathfrak{B}^{^{\!\vee}}}

\newcommand{\cess}{\mathfrak{B}^{^{\!\vee\!\!\!\vee}}}

\newcommand{\pos}{\mathfrak{B}^{^{\!+}}}
\newcommand{\parv}{\mathfrak{B}^{^{\!\oplus}}}
\newcommand{\van}{\mathfrak{B}^{^{\!0}}}

\newcommand{\ce}{\slashed{E}} 
\newcommand{\tcut}[1]{C^\circ_{#1}} 
\newcommand{\tc}{\gf{T}^\circ} 
\newcommand{\mmC}[1]{\text{\textit{\u{C}}}_{\!#1}} 

\newcommand{\ent}{{\sf S}}  
\newcommand{\mi}{{\sf I}}  
\newcommand{\pmi}{\mathcal{P}}   
\newcommand{\face}{\mathscr{F}}   

\newcommand{\Svec}{\vec{\ent}}  

\newcommand{\gf}[1]{{\mathrm{#1}}}
\newcommand{\hp}{\gf{H}_{\mathcal{P}}}  
\newcommand{\lhp}{\gf{L}_{\mathcal{P}}}  

\SetKwComment{Comment}{** }{ **} 
\SetKwInOut{Input}{input}\SetKwInOut{Output}{output}

\definecolor{PminusEcol}{rgb}{1,0.8,0.8}
\definecolor{EminusMcol}{rgb}{1,0.85,0.65}
\definecolor{Mcol}{rgb}{1,0.97,0.5}
\definecolor{Qcol}{rgb}{0.7,1,0.9}
\definecolor{Vcol}{rgb}{0.83,0.83,1}
\definecolor{EminusPcol}{rgb}{0.8,1,0.6}
\definecolor{Ccol}{rgb}{0.9,0.9,0.9}
\definecolor{bulkcol}{rgb}{0.8,0.8,0.8}

\title{Necessary and sufficient conditions for entropy vector realizability by holographic simple tree graph models}

\author[a]{Veronika E. Hubeny,}
\emailAdd{veronika@physics.ucdavis.edu}

\author[b]{Massimiliano Rota}
\emailAdd{max.rota@bristol.ac.uk}

\affiliation[a]{Center for Quantum Mathematics and Physics (QMAP)\\ 
Department of Physics \& Astronomy, University of California, Davis, CA 95616 USA}
\affiliation[b]{School of Mathematics, University of Bristol,
Woodland Road, Bristol, BS8 1UG UK}

\abstract{We prove that the ``chordality condition'', which was established in arXiv:2412.18018 as a necessary condition for an entropy vector to be realizable by a holographic simple tree graph model, is also sufficient. The proof is constructive, demonstrating that the algorithm introduced in arXiv:2512.18702 for constructing a simple tree graph model realization of a given entropy vector that satisfies this condition always succeeds. We emphasize that these results hold for an arbitrary number of parties, and, given that any entropy vector realizable by a holographic graph model can also be realized, at least approximately, by a stabilizer state, they highlight how techniques originally developed in holography can provide broad insights into entanglement and information theory more generally, and in particular, into the structure of the stabilizer and quantum entropy cones. Moreover, if the strong form of the conjecture from arXiv:2204.00075 holds, namely, if all holographic entropy vectors can be realized by (not necessarily simple) tree graph models, then the result of this work demonstrates that the essential data that encodes the structure of the holographic entropy cone for an arbitrary number of parties, is the set of ``chordal'' extreme rays of the subadditivity cone.
}

\begin{document}
 

\maketitle

\section{Introduction}

For a fixed number of parties $\N$, and a collection of $\D=2^{\N}-1$ non-negative real numbers, is there a Hilbert space, and a density matrix, such that these numbers are the von Neumann entropies of the reduced density matrices for all possible (non-empty) collections of parties? This question is of paramount importance in information theory, both classical and quantum, where entropic constraints play a key role in establishing fundamental bounds on information processing. 

For $\N=2,3$ the answer to this question was found in \cite{641556,1193790}, which introduced the notions of classical and quantum \textit{entropy cones}. A collection of non-negative real numbers is the collection of entropies of a classical probability distribution, or a density matrix, if and only if it satisfies all basic entropy inequalities: \textit{subadditivity} (SA), \textit{strong subadditivity} (SSA), and in the classical case, \textit{monotonicity}.\footnote{\,In the quantum case there are additional non-redundant inequalities that specify the entropy cone, namely the Araki-Lieb inequality, and weak monotonicity. But since these can be obtained (respectively) from SA and SSA by introducing purifications, they can be interpreted as particular instances of these fundamental inequalities.}$^,$\footnote{\,More precisely, it might be possible to realize a given collection of entropies that obeys all these inequalities only approximately. This subtlety however is immaterial for the purposes of this work, which is mainly concerned with the construction of holographic graph models, rather than explicit quantum states. \label{ft:approx}} For $\N\geq 4$ on the other hand, the detailed structure of the classical and quantum entropy cones is not known, and an answer to the general question presented above is an open problem.\footnote{\,In the classical case, already at $\N=4$ it is known that there exist infinitely many non-redundant linear inequalities \cite{2007:Matus}, and the classical entropy cone is therefore not polyhedral.}

The same question can also be posed for restricted classes of states. For instance, in the case of stabilizer states, the structure of the entropy cone is well-understood for $\N = 4$ \cite{Linden:2013kal}. Apart from intrinsic interest, particularly given the wide-ranging applications of stabilizer states in quantum information theory, knowledge of the stabilizer entropy cone also provides an interesting inner bound to the full quantum entropy cone. 

Another important class of quantum states, which will be the focus of this work, is that of ``geometric states'' in the context of gauge-gravity duality \cite{Maldacena:1997re,Witten:1998qj,Gubser:1998bc}. These are the states of holographic CFTs that correspond to classical bulk geometries, for which the Ryu-Takayanagi \cite{Ryu:2006bv} formula computes the von Neumann entropy of boundary spatial subsystems.\footnote{\,In this work, we focus on static spacetimes.} The entropy cone associated with these states was introduced in \cite{Bao:2015bfa} and is referred to as the \textit{holographic entropy cone} (HEC). 
It was already shown in \cite{Bao:2015bfa} that the HEC can equivalently (and conveniently) be defined as the set of entropy vectors realized by certain ``graph models'', and in \cite{hayden2016holographic} that it is contained within the $\N$-party stabilizer entropy cone.

The study of the HEC is of course also important in the context of quantum gravity, where it remains a longstanding question how a bulk classical geometry is encoded in boundary data. While the HEC is completely known up to $\N=5$ \cite{Bao:2015bfa,Cuenca:2019uzx}, and several results have more recently been obtained for larger values of $\N$ \cite{Hernandez-Cuenca:2023iqh,Czech:2024rco}, a deep understanding of its general structure and the interpretation of the inequalities remains elusive.\footnote{\,See \cite{Czech:2025jnw} for first steps in this direction.} A conjecture about its structure in relation to the outer bound specified by the ``subadditivity cone'' was proposed in \cite{Hernandez-Cuenca:2022pst}, but even testing this conjecture is a challenging task.

For these reasons, new techniques for building graph models of given entropy vectors and eventually detecting unrealizability independently of knowledge of the inequalities were recently introduced in \cite{graph-construction}. In particular, given an arbitrary entropy vector that satisfies SA and SSA, it was shown in \cite{Hubeny:2024fjn} that additionally a certain ``chordality condition'' is necessary for the existence of a realization by a ``simple tree'' holographic graph model. For entropy vectors that obey these conditions, \cite{graph-construction} recently introduced an algorithm for the efficient construction of a candidate graph model. The goal of this work is to prove that these conditions are also sufficient, and that the algorithm introduced in \cite{graph-construction} always succeeds. Specifically, we will prove the following theorem.

\begin{thm}
\label{thm:main}
    An entropy vector that satisfies \emph{SA} and \emph{SSA} can be realized by a holographic graph model which is a simple forest if and only if the line graph of its correlation hypergraph is a chordal graph. 
\end{thm}

To make this work self-contained, in \S\ref{sec:preliminaries} we will briefly review all necessary definitions from earlier work: the correlation hypergraph from \cite{Hubeny:2024fjn} in \S\ref{subsec:basics}, the fundamentals of holographic graph models in \S\ref{subsec:graphs}, and the algorithm from \cite{graph-construction} in \S\ref{subsec:algorithm}. Note that \Cref{thm:main} allows for the possibility that the graph realizing a given entropy vector is a forest rather than a tree. This is related to the fact that, under certain relatively special conditions, the problem of finding a graph model for a given entropy vector can be reduced to the analogous problem for entropy vectors associated to fewer party systems (in this case, the trees that make up the forest). This reduction, which was discussed extensively in \cite{graph-construction}, will be briefly reviewed in \S\ref{subsec:algorithm}, which accordingly introduces a reduced version of \Cref{thm:main} tailored to that setting. The proof of this reduced version will be the goal of \S\ref{sec:proof} and will complete the argument for \Cref{thm:main}. We conclude this work in \S\ref{sec:discussion} with a few comments on the implications of this result and open questions for future investigations.

\section{Preliminaries}
\label{sec:preliminaries}

In this section we briefly review the definitions which are essential to make this work self-contained. For more details about the concepts reviewed in this subsection, we refer the reader to the original literature.

\subsection{KC-PMIs and the correlation hypergraph}
\label{subsec:basics}

For a fixed number of parties $\N$, we denote a party by an integer $\ell\in[\N]$, where $[\N]=\{1,2,\ldots,\N\}$. The ancillary party that purifies an $\N$-party density matrix will conventionally be denoted by 0, and the set of $\N+1$ parties, including the purifier, by $\nsp=\{0,1,\ldots,\N\}$. A collection of $\D=2^\N-1$ real numbers, corresponding to the non-empty subsets of $[\N]$, with a conventional order that for the examples presented below we choose to be lexicographic, is called an \textit{entropy vector}, and the vector space $\mathbb{R}^{\D}$ is called \textit{entropy space}. For convenience, the instances of the subadditivity inequality (SA) will be written in terms of the mutual information as $\mi(\X:\Y)\geq 0$, where $\X,\Y$ are subsets of $\nsp$. We will always implicitly assume that if a term in this inequality includes the purifier, the term is replaced by the entropy of its complement. The set of all SA instances therefore, by definition, also includes all instances of the Araki-Lieb inequality. The $\N$-party \textit{subadditivity cone} (SAC$_\N$) is the polyhedral cone specified by all instances of subadditivity.

Given a face $\face$ of the SAC$_\N$, and an arbitrary entropy vector $\Svec$ in the interior of $\face$, the set of MI instances that vanish for $\Svec$ is called the \textit{pattern of marginal independence} (PMI) of $\Svec$, or equivalently, the PMI of $\face$ (since it does not depend on the specific choice of $\Svec \in \text{int}\face$). A PMI is said to be SSA-compatible if there exists at least one entropy vector in the interior of the corresponding face that satisfies all instances of SSA, $\mi(\X:\Y\Z)\geq \mi(\X:\Y)$. It is easy to see that any SSA-compatible PMI $\pmi$ satisfies the following \textit{Klein's condition}
\begin{equation}
    \mi(\X:\Y)\in\pmi \implies \mi(\X':\Y')\in\pmi,
\end{equation}
for every $\X'\subseteq\X$ and $\Y'\subseteq\Y$, or $\X'\subseteq\Y$ and $\Y'\subseteq\X$. The converse (regarding SSA) is in general not true, and a PMI that satisfies this weaker condition is called a KC-PMI.

For fixed $\N$, we denote by $\mgs_\N$ the set of all MI instances, and we define the MI-\textit{poset} as the partially ordered set $(\mgs_\N,\preceq)$, where the partial order is
\begin{equation}
    \mi(\X':\Y')\preceq \mi(\X:\Y)\iff \text{$\X'\subseteq\X$ and $\Y'\subseteq\Y$, or $\X'\subseteq\Y$ and $\Y'\subseteq\X$}.
\end{equation}
A PMI $\pmi$ is then a KC-PMI if and only if it is a down-set in the MI-poset. Not all down-sets, however, are PMIs, due to the linear dependence among the MI instances, and SA.

For an arbitrary subsystem $\X\subseteq\nsp$ with $|\X|\geq 2$, the \textit{$\bset$ of} $\X$, denoted by $\bs{\X}$, is the set of MI instances whose arguments form a bipartition of $\X$, i.e.,
\begin{equation}
    \bs{\X}\coloneq \{\mi(\Y:\Z)\in\mgs_\N|\; \Y\cup\Z=\X \}
\end{equation}
The set of all $\bsets$ forms a partition of $\mgs_\N$.

Given an arbitrary down-set $\ds$ of the MI-poset, we denote by $\cds$ its complement in $\mgs_\N$. A $\bset$ is then said to be \textit{positive} if it is a subset of $\cds$, \textit{vanishing} if it is a subset of $\ds$, and \textit{partial} otherwise. Furthermore, it is said to be \textit{essential} if it contains at least one element of the antichain $\ac$ whose up-set is $\cds$, and \textit{completely essential} if it is a subset of $\ac$. For a given $\ds$, the set of positive $\bsets$ is denoted by $\pos$, the set of partial $\bsets$ by $\parv$, the set of vanishing $\bsets$ by $\van$, the set of essential $\bsets$ by $\ess$, and the set of completely essential $\bsets$ by $\cess$. 

While an arbitrary down-set can have a $\bset$ that is simultaneously partial and essential, it was shown in \cite{Hubeny:2024fjn} that for any $\N$ and any KC-PMI $\pmi$, \textit{every essential $\bset$ is positive}, implying that $\pmi$ is uniquely determined by the set of its positive $\bsets$. In particular, this allows for a suggestive representation of $\pmi$ in terms of a hypergraph $\hp=(V,E)$, called the \textit{correlation hypergraph} of $\pmi$, whose vertices correspond to the $\N+1$ parties, and hyperedges to subsystems whose $\bset$ is positive, i.e., 
\begin{equation}
 V = \{v_\ell|\, \ell\in \nsp\} \quad E= \{h_\X\subseteq V|\, \bs{\X}\; \text{is positive}\},
\end{equation}
where in $h_\X$, $\X=\{\ell\in\nsp|\, v_\ell\in h_\X\}$.

For an arbitrary subsystem $\X\subseteq \nsp$, with $|\X|\geq 2$, whose $\bset$ is either partial, vanishing, or essential,
the $\Gamma$-\textit{partition} of $\X$ is the partition\footnote{\,To simplify the notation, when considering a set $\mathcal{A}$ with a single element $a$ we simply write the element, i.e., we write $a$ instead of $\{a\}$.}
\begin{equation}
\label{eq:gamma-partition}
    \Gamma(\X)=\{\X_1,\ldots,\X_n,\ell_{1},\ldots\ell_{\tilde{n}}\}
\end{equation}
where each $\X_i$, with $i\in[n]$, is a maximal (with respect to inclusion) proper subset of $\X$ whose $\bset$ is positive. The fact that such collection of subsystems $\X_i$ is indeed a partition of a subset of $\X$ is not trivial, and was shown in \cite{Hubeny:2024fjn}. In \eqref{eq:gamma-partition}, the parties $\ell_1,\ldots,\ell_{\tilde{n}}$ are simply added, if necessary, to complete the partition. Given a KC-PMI $\pmi$, and subsystem $\X$ whose $\bset$ is partial or vanishing, it is possible to immediately determine from $\Gamma(\X)$ which elements of $\bs{\X}$ belong to $\pmi$: They are the MI instances $\mi(\Y:\Z)$ such that the bipartition $\{\Y,\Z\}$ of $\X$ is a coarser\footnote{\,Not necessarily strictly, i.e., it could be that $\{\Y,\Z\}=\Gamma(\X)$.} partition of $\X$ compared to $\Gamma(\X)$ (in the lattice of partitions of $\X$). For a subsystem $\X$ whose $\bset$ is essential, the same criterion determines which MI instances in $\bs{\X}$ belong to the antichain $\ac$.

\subsection{Holographic graph models}
\label{subsec:graphs}

An $\N$-party holographic graph model \cite{Bao:2015bfa} is a weighted graph $\gf{G}=(V,E)$, with non-negative edge weights, and a specification of a subset $\partial V\subseteq V$ of vertices called \textit{boundary vertices} which are labeled by the parties in $\nsp$. Each boundary vertex is labeled by exactly one party, and each party $\ell\in\nsp$ labels at least one boundary vertex. Vertices in $V\setminus \partial V$ are referred to as \textit{bulk vertices.} A holographic graph model is said to be \textit{simple} if each boundary vertex is labeled by a distinct party. A holographic \textit{simple tree} graph model is a simple holographic graph model with tree topology, i.e., a connected graph without cycles. Similarly, a \textit{simple forest} (cf., \Cref{thm:main}) is a simple holographic graph model with the topology of a forest, i.e., a graph without cycles which is not necessarily connected.

An arbitrary subset $C\subseteq V$ is called a \textit{cut}. For a subsystem $\X\subseteq\nsp$, an \textit{$\X$-cut}, denoted by $C_\X$, is a cut such that the set of boundary vertices in $C_\X$ is the set of \textit{all} boundary vertices labeled by the parties in $\X$. Given a cut, the set of \textit{cut edges}, denoted by $\ce(C)$, is the set of edges with one endpoint in the cut, and one in the complement $C_\X^\complement=V\setminus V_\X$, i.e.,
\begin{equation}
    \ce(C) \coloneqq \{\{v,v'\}\in E,\; v\in C,\; v'\in C^\complement\}.
\end{equation}
The \textit{cost} of an $\X$-cut is then defined as the sum of the weights of its cut edges, 
    \begin{equation}
        \norm{C} = \sum_{e\in\ce(C)} w(e),
    \end{equation}
where $w(e)$ is the weight of the edge $e$. 

For a subsystem $\X$, a \textit{min-cut for $\X$} is an $\X$-cut with minimum cost. Min-cuts are in general non-unique, but if for a subsystem $\X$ there are multiple min-cuts, the minimal one, with respect to inclusion, is unique and it is called the \textit{minimal min-cut} for $\X$ \cite{Avis:2021xnz}. For a given holographic graph model and subsystem $\X$, the \textit{entropy} of $\X$ is defined as the cost of a min-cut for $\X$. The collection of the values of the entropy for the various non-empty subsets of $[\N]$, conventionally ordered, is the entropy vector $\Svec$ of the graph model.

\subsection{Algorithm for simple tree graph models}
\label{subsec:algorithm}

Suppose that for some number of parties $\N$ we are given an entropy vector $\Svec$ that obeys all instances of SA and SSA, and we want to determine whether it can be realized by a holographic graph model, and in that case, find a graph realization. Since $\Svec$ obeys SA, its PMI $\pmi$ is well defined, and since it obeys SSA, $\pmi$ is a KC-PMI and can be described in terms of the correlation hypergraph $\hp$.

It is easy to see \cite{Hubeny:2024fjn} that $\hp$ is connected if and only if $\Svec$ has no vanishing entries. Based on this observation, and the fact that a similar property holds for graph models, \cite[Theorem 9]{graph-construction} then showed that if $\Svec$ \textit{does} have vanishing entries, then the problem of finding a graph model for $\Svec$ can easily be reduced to that of finding graph models for a collection of entropy vectors for fewer party systems which do not have any vanishing entry, and can be obtained directly from $\Svec$ (importantly, without knowing any graph realization or even if such a realization exists).

In what follows, a key role will be played by the \textit{line graph}\footnote{\,The line graph of a hypergraph is the intersection graph of its set of hyperedges.} of $\hp$, denoted by $\lhp$. For any max-clique $Q$ of $\lhp$, we define the set
\begin{equation}
Q^\cap = \bigcap_{h_\X\in Q} h_\X.
\end{equation}
It was shown in \cite[Theorem 9]{Hubeny:2024fjn} that for any KC-PMI, the cardinality of $Q^\cap$ satisfies the stringent bound $|Q^\cap|\leq 2$, and that it is strictly positive only if the correlation hypergraph has a particular structure. Using this result, \cite[Theorem 10]{graph-construction}  then showed that, similarly to the discussion above, the problem of finding a graph model for $\Svec$ can be reduced to that of finding graph models for a collection of entropy vectors for fewer party systems, which can be immediately derived from $\Svec$ and which obey $|Q^\cap|=0$ for all max-cliques of their line graphs.

Entropy vectors for which either of these reductions is possible are called ``reducible'', while any other entropy vector (obeying SA and SSA) is said to be ``irreducible''. Specifically, we introduce the following definition.

\begin{defi}
    An entropy vector $\Svec$ is said to be irreducible if it satisfies all instances of \emph{SA} and \emph{SSA}, it has no vanishing entries, and $Q^\cap=\varnothing$ for all max-cliques of $\lhp$. 
\end{defi}

The proof of \Cref{thm:main} then follows immediately from these results from \cite{graph-construction} and the proof of the analogous statement for irreducible entropy vectors. Specifically, in \S\ref{sec:proof} we will prove the following theorem.

\begin{thm}
\label{thm:main-reduced}
    An irreducible entropy vector can be realized by a holographic simple tree graph model if and only if the line graph of its correlation hypergraph is a chordal graph. 
\end{thm}

A necessary condition for the realizability of an entropy vector $\Svec$ obeying SA and SSA by a holographic simple tree graph model was proven in \cite{Hubeny:2024fjn}, and it states that such a mode exists \textit{only if} $\lhp$ is a \textit{chordal graph}.\footnote{\,A graph is said to be chordal if it does not contain a cycle of length four or more without a chord.} Based on this condition, \cite{graph-construction} then proposed an algorithm for the construction of a simple tree graph model for any irreducible entropy vector $\Svec$ whose line graph is chordal. (We refer to the corresponding KC-PMI as chordal KC-PMI.) The algorithm is shown in Algorithm~\ref{alg:reconstruction}, and we refer the reader to \cite{graph-construction} for a discussion of the logic behind it and how it was developed. Here we briefly review the various steps mainly to clarify the definition and notation. (The numbers refer to those in the Algorithm.)

\begin{algorithm}[t]
\caption{Construction of simple tree graph model $\gf{T}$ realizing an entropy vector $\Svec$ satisfying the assumptions described in main text.}
\label{alg:reconstruction}
\BlankLine
\Input{an irreducible entropy vector $\Svec$ whose PMI is a chordal KC-PMI}
\Output{a simple tree graph model $\gf{T}$ realizing $\Svec$}
\BlankLine
$\pmi \leftarrow \text{the PMI of}\; \Svec$\;
$\pos \leftarrow \text{the set of positive $\bsets$ of}\; \pmi$\;
$\lhp \leftarrow \text{the intersection graph of}\; \pos$\;
\nl $\gf{K}_\pmi \leftarrow \text{the weighted clique graph of}\; \lhp$\;
\nl $\widetilde{\gf{T}} \leftarrow \text{a choice of a maximum spanning tree of}\; \gf{K}_\pmi$\;
$\gf{T} \leftarrow \widetilde{\gf{T}}$\;
\nl \For{each $\ell\in\nsp$}
{
    construct the clique $Q^\ell$ in $\lhp$\;
    find the unique max-clique $Q$ of $\lhp$ that contains $Q^\ell$\; 
    add to $\gf{T}$ a vertex $v_\ell$ for the party $\ell$\;
    add to $\gf{T}$ an edge connecting $v_\ell$ to the vertex $v_Q$ corresponding to $Q$\;
    }
\nl \For{each edge e of $\gf{T}$}
{
    $\Y\leftarrow$ the element of $\langle\X,\X^\complement\rangle_e$ that does not contain the purifier\;
    \eIf{$\Y\neq \varnothing$}
    {$w(e) \leftarrow \ent_{\Y}$\;}
    {$w(e) \leftarrow 0$\;}
    }
\end{algorithm}

\begin{enumerate}[label={\scriptsize \textbf{\arabic*})}]
\item The \textit{clique graph} $\gf{K}_\pmi$ of a graph $\gf{G}$ is the intersection graph of the set of max-cliques of $\gf{G}$. Its weighted version is obtained by adding to each edge $e$ a weight equal to the cardinality of the intersection of the max-cliques of $\gf{G}$ corresponding to the endpoints of $e$.

\item A \textit{maximum spanning tree} of a weighted graph $\gf{G}$ is a spanning tree $\widetilde{\gf{T}}$ such that the sum of the weights of the edges of $\gf{G}$ which are also in $\widetilde{\gf{T}}$ has the maximum value among all spanning trees. In general, a maximum spanning tree is not unique, and the algorithm simply prescribes to make an arbitrary choice.

Since (by assumption) $\lhp$ is a chordal graph, a maximum spanning tree $\widetilde{\gf{T}}$ of its weighted clique graph is a \textit{clique tree representation} of $\lhp$ \cite{book:intersection_graphs}. This is a tree whose vertex set is the set of max-cliques of $\lhp$, and such that for each vertex $v$ of $\lhp$, the set of max-cliques that contain $v$ induces a subtree of $\widetilde{\gf{T}}$.

\item The loop constructs a ``topological'' simple tree graph model $\gf{T}$ from $\widetilde{\gf{T}}$ by adding leaves for the $\N+1$ boundary vertices. For a party $\ell$, the clique $Q^\ell$ of $\lhp$ is the set of all hyperedges of $\hp$ which contain the vertex $v_\ell$. It was shown in \cite[Theorem 9]{Hubeny:2024fjn} that $Q^\ell$ is contained in a unique max-clique $Q$ of $\lhp$, and this uniqueness guarantees that there is no ambiguity in the construction.\footnote{\,This is not the case when $Q^\ell$ is empty, but this is possibility is excluded here by the assumption that $\Svec$ is irreducible.} 

\item The loop constructs a graph model by adding weights to the edges of $\gf{T}$. For each edge $e$, we denote by $\langle\X,\X^\complement\rangle_e$ the pair of subsystems corresponding to the boundary vertices of the two trees obtained from $\gf{T}$ after deleting $e$, and by $\Y$ the element of this pair that does not contain the purifier. The weight $w(e)$ is then chosen to be the component $\ent_\Y$ of $\Svec$. Notice that, a priori, it is possible that, for some edge $e$, the two subsystems are the empty set and the full system $\nsp$. In this case, there is no corresponding component in $\Svec$ and the algorithm assigns vanishing weight to $e$. However, we will see below that this situation never occurs.

\end{enumerate}

\section{Proof of \Cref{thm:main-reduced}}
\label{sec:proof}

Having reviewed the necessary background, we now begin the proof of \Cref{thm:main-reduced}. We assume that for an arbitrary number of parties $\N$ we are given an entropy vector $\Svec$ which is irreducible and has a chordal line graph. It is important to note that under these assumptions, Algorithm~\ref{alg:reconstruction} always gives a \textit{candidate} simple tree graph model $\gf{T}$ for $\Svec$ (which a priori is not unique because of the choice of $\widetilde{\gf{T}}$). What is far from obvious however, is that $\gf{T}$ is indeed a graph model realization of $\Svec$, or in other words, that the entropy vector computed from $\gf{T}$ using the standard min-cut prescription reviewed in \S\ref{subsec:graphs}, is indeed $\Svec$. Our strategy to prove \Cref{thm:main-reduced} will be to prove the following result, which highlights the constructive nature of our approach.

\begin{thm}
\label{thm:alg-suff}
    For any irreducible entropy vector $\Svec$, any holographic simple tree graph model $\gf{T}$ resulting from Algorithm~\ref{alg:reconstruction} is a realization of $\Svec$.
\end{thm}

We now summarize the logic of the strategy that we will follow to prove \Cref{thm:alg-suff}. Note that while Algorithm~\ref{alg:reconstruction} gives a candidate simple tree graph model $\gf{T}$ for $\Svec$, it does not immediately specify what the min-cuts (and in particular the minimal ones) for the various subsystems are. The first step of the proof will be to make a guess for the minimal min-cuts for the various subsystems, which is based on the intuition behind the algorithm presented in \cite{graph-construction}. Specifically, for each subsystem $\X\subseteq\nsp$, we make a choice\footnote{\,Notice that we do not use the notation $\mmC{\X}$ introduced above for minimal min-cuts. The reason for this is that showing that $\tcut{\X}$ is indeed a min-cut for $\X$ will be a key element of the proof.} of an $\X$-cut, denoted by $\tcut{\X}$, as follows.

\begin{enumerate}
    
    \item For each $\X$ such that $\bs{\X}\in\pos$ we choose the $\X$-cut
    \begin{equation}
    \label{eq:cut-pos-bs}
        \tcut{\X}= \{v_\ell|\, \ell\in\X\} \cup \{v_Q|\, \X\in Q\}.
    \end{equation}
    Since we are assuming that $\lhp$ is connected (by the fact that $\Svec$ is irreducible), $\X=\nsp$ is a special case of this class, and $\tcut{\nsp}$ is the whole vertex set of $\gf{T}$.
    
    \item For each single party $\ell$ we choose the $\X$-cut 
    \begin{equation}
    \label{eq:cut-single-party}
          \tcut{\ell}=\{v_\ell\}. 
    \end{equation}

    \item For each $\X$ such that $\bs{\X}\in\parv\cup\van$ and $|\X|\geq 2$, we choose the $\X$-cut
    \begin{equation}
    \label{eq:cut-partial-bs}
        \tcut{\X}=\bigcup_{\Y\in\Gamma(\X)}\tcut{\Y}.
    \end{equation}
    In the particular case where $\bs{\X}\in\van$, this simply gives $\tcut{\X}=\{v_\ell|\, \ell\in\X\}$. 

    \item In the trivial case $\X=\varnothing$, we specify the cut
    \begin{equation}
    \label{eq:cut-empty}
        \tcut{\varnothing}=\varnothing.
    \end{equation}

\end{enumerate}

Since in holographic graph models the purifier appear explicitly, for the purpose of the proof of \Cref{thm:alg-suff} it will be convenient to work in the vector space $\mathbb{R}^{2^{\N+1}}$, rather in conventional entropy space. To avoid introducing unnecessary new notation, in what follows we will also denote by $\Svec$ the entropy vector in this ``extended'' entropy space, whose components are, for all subsystems $\X$ which do not include the purifier the same as in the original entropy vector, for their complements (which therefore do contain the purifier)
\begin{equation}
\label{eq:S-extension1}
    \ent_{\X^\complement} = \ent_{\X}\, ,
\end{equation}
and finally
\begin{equation}
\label{eq:S-extension2}
    \ent_{\nsp}=\ent_{\varnothing}=0,
\end{equation}
which are therefore the only vanishing components.

Given the choice of cuts above, and a simple tree graph model $\gf{T}$ resulting from Algorithm~\ref{alg:reconstruction}, we then construct a vector $\vec{{\sf C}}$ in the same space as the new $\Svec$ by setting, for each $\X\subseteq \nsp$,
\begin{equation}
    {\sf C}_\X = \norm{\tcut{\X}}. 
\end{equation}
Notice that for $\vec{{\sf C}}$ it is a priori not clear whether ${\sf C}_\X={\sf C}_{\X^\complement}$, which will be a central aspect of the discussion.  To prove \Cref{thm:alg-suff}, it then suffices to first show that ${\sf C}_\X=\ent_\X$ for all subsystems $\X\subseteq \nsp$, which is the goal of \S\ref{subsec:cost-eq}, and then that each $\tcut{\X}$ is a min-cut, which is the goal of \S\ref{subsec:min-cut}. 

To simplify the presentation, in the statements of the various lemmas and theorems below we will always implicitly assume—and therefore not state explicitly—that one is given an irreducible\footnote{\,We remind the reader that the notion of irreducibility, in particular the fact that $\Svec$ has no vanishing components, is formulated in standard (rather than ``extended'') entropy space.} entropy vector $\Svec$ and has obtained a simple tree graph model $\gf{T}$ from Algorithm~\ref{alg:reconstruction} after a choice of $\widetilde{\gf{T}}$. The dependence of $\gf{T}$ on $\widetilde{\gf{T}}$ will be kept implicit, since, as we will see, it is immaterial. To simplify the notation, given two subsystems $\X,\Y$ we will often write $\X\Y$ instead of $\X\cup\Y$.

\subsection{The cost of the cuts $\tcut{\X}$ on $\gf{T}$ reproduces the components of $\Svec$}
\label{subsec:cost-eq}

We begin by briefly reviewing a simple result from \cite{Hubeny:2024fjn}, which  for an entropy vector $\Svec$ that obeys SA and SSA, translates part of the structure of the correlation hypergraph into linear relations among certain components of $\Svec$. Specifically,\footnote{\,The result of \cite{Hubeny:2024fjn} is slightly more general, since it only requires that the PMI of $\Svec$ is a KC-PMI; a requirement that in general is strictly weaker than SSA.} for each subsystem $\X$ such that $|\X|\geq 2$ and $\bs{\X}$ is \textit{not} positive\footnote{\,This relation is immediate to derive from the set of vanishing MI instances in $\bs{\X}$, which are made manifest by $\Gamma(\X)$.}
\begin{equation}
\label{eq:entvec-cost-non-pos-bsets}
    \ent_\X = \sum_{\Y\in\Gamma(\X)} \ent_\Y.
\end{equation}
Our first goal in this subsection will be to show that exactly the same relations are satisfied by the costs of the cuts $\tcut{\X}$, i.e.,
\begin{equation}
\label{eq:cost-non-pos-bsets}
    \norm{\tcut{\X}}=\sum_{\Y\,\in\,\Gamma(\X)} \norm{\tcut{\Y}}  \qquad \forall\, \X\;\; \text{such that}\;\; \bs{\X}\notin\pos.
\end{equation}

The first step is to show that, by construction, the cuts $\tcut{\X}$ satisfy two general properties of minimal min-cuts \cite{Avis:2021xnz}.

\begin{lemma}
\label{lem:fund-cut-properties}
    For any pair of subsystems $\X,\Y$:
    \begin{enumerate}[label={\emph{\footnotesize \roman*)}}]
        \item 
        $ \quad
             \X\cap\Y=\varnothing  \quad \iff \quad \tcut{\X}\cap\tcut{\Y}=\varnothing
        $
        \item 
        $ \quad
            \X\subseteq\Y \quad  \implies \quad \tcut{\X}\subseteq\tcut{\Y}
        $
    \end{enumerate}
\end{lemma}
\begin{proof}
    If either $\X$ or $\Y$ is empty, by \eqref{eq:cut-empty} all these statements are trivial, and we assume $\text{min}(|\X|,|\Y|)\geq 1$.

    (i) $(\Leftarrow)$ Follows trivially from the definition of the cuts $\tcut{\X}$ and $\tcut{\Y}$. 
    
    (i) $(\Rightarrow)$ Suppose that $\X\cap\Y=\varnothing$ and consider the definition of the cuts given in \eqref{eq:cut-partial-bs}. We rewrite
    \begin{align}
    \label{eq:cuts-intersection}
        \tcut{\X}\cap\tcut{\Y} & = \left(\bigcup_{\X'\in\Gamma(\X)}\tcut{\X'}\right) \cap \left(\bigcup_{\Y'\in\Gamma(\Y)}\tcut{\Y'}\right) \nonumber\\
        & = \bigcup_{\substack{\X'\in\Gamma(\X)\\ \Y'\in\Gamma(\Y)}} \tcut{\X'} \cap \tcut{\Y'} \, ,
    \end{align}
    where we have simply used distributivity of intersection over union. In \eqref{eq:cuts-intersection}, if $\text{min}(|\X'|,|\Y'|)=1$, then $\tcut{\X'} \cap \tcut{\Y'}=\varnothing$ by the definition of the cuts for single parties (cf., \eqref{eq:cut-single-party}) and the assumption $\X\cap\Y=\varnothing$. 

    Suppose then that $\text{min}(|\X'|,|\Y'|)>1$, and therefore that $\bs{\X'}$ and $\bs{\Y'}$ are both positive. In this case these $\bsets$ are both vertices of $\lhp$. Since we are assuming $\X\cap\Y=\varnothing$ we also have $\X'\cap\Y'=\varnothing$, and any max-clique of $\lhp$ either contains $\bs{X'}$, or $\bs{Y'}$, or neither of them, but it cannot contain both. Therefore, from the definition of $\tcut{\X'}$ and $\tcut{\Y'}$ in \eqref{eq:cut-pos-bs}, it follows that $\tcut{\X'}\cap\tcut{\Y'}=\varnothing$. Since we have shown that $\tcut{\X'}\cap\tcut{\Y'}=\varnothing$ for any $\X'\in\Gamma(X)$ and $\Y'\in\Gamma(Y)$, we obtain $\tcut{\X}\cap\tcut{\Y}=\varnothing$. 

    (ii) By definition of $\tcut{\X}$ and $\tcut{\Y}$, and the assumption $\X\subseteq\Y$, every boundary vertex in $\tcut{\X}$ is also in $\tcut{\Y}$, and we need to show the inclusion for bulk vertices. We will distinguish four cases, depending on whether $\bs{\X}$ and $\bs{\Y}$ are positive or not. If they are both positive, they are both vertices of $\lhp$, and any max-clique of $\lhp$ containing the vertex $\bs{\X}$ also contains $\bs{\Y}$, proving the inclusion. If $\bs{\X}$ is positive but $\bs{\Y}$ is not, then (by the definition of $\Gamma$-partitions) there is some $\Z\in\Gamma(\Y)$ such that $\bs{\Z}$ is positive and $\X\subseteq\Z$, and the inclusion follows from the same argument as above, now applied to $\X$ and $\Z$. The same logic also applies to the case where $\bs{\X}$ is not positive and $\bs{\Y}$ is, since in this case, for any component $\Z$ of $\Gamma(\X)$, we have $\Z\subseteq\Y$. The last case, where neither $\bs{\X}$ nor $\bs{\Y}$ is positive, is again analogous, as it suffices to notice that each element of $\Gamma(\X)$ is a subset of an element of $\Gamma(\Y)$. 
\end{proof}

Notice that in the proof of (ii) of \Cref{lem:fund-cut-properties}, one could have argued that if $\X\subseteq\Y$ then $\X\cap\Y^\complement=\varnothing$, implying that $\tcut{\X}\cap\tcut{\Y^\complement}=\varnothing$. However, it is important to notice that this does not imply that $\tcut{\X}\subseteq\tcut{\Y}$, because in general it is not the case that $\tcut{\Y^\complement}=(\tcut{\Y})^\complement$. In fact, in general we only have $\tcut{\Y}\cup\tcut{\Y^\complement}\subseteq V(\gf{T})$, but the inclusion can be strict. For example, for the $\N=3$ star graph model realizing the ``perfect state'' \cite{Bao:2015bfa}, for any $\Z$ with $|\Z|=2$, neither $\tcut{\Z}$ nor $\tcut{\Z^\complement}$ contains the bulk vertex.\footnote{\,In fact, the relation $\tcut{\Z^\complement}=(\tcut{\Z})^\complement$ must always be false if the min-cut for $\Z$ is not unique, since the minimal min-cut is complementary to the maximal min-cut of the complementary subsystem.} 

We now introduce a convenient notation that we will use extensively below. For a choice of $\widetilde{\gf{T}}$ and subsystem $\X$, we denote by $\tc_\X$ the subgraph of $\gf{T}$ induced by $\tcut{\X}$. Notice that for every subsystem $\X$, the definition of $\tcut{\X}$ is independent from the choice of $\widetilde{\gf{T}}$, but the graph $\tc_\X$ does depend on this choice. Nevertheless, as usual, we will keep such dependence implicit in what follows. Furthermore, we stress that for an arbitrary choice of $\X$, the subgraph $\tc_\X$ is not necessarily a tree, since it can be disconnected.

\begin{lemma}
\label{lem:positive-bs-connected}
    For any subsystem $\X$ whose $\bset$ is positive, the subgraph $\tc_\X$ of $\gf{T}$ induced by the cut $\tcut{\X}$ is a subtree of $\gf{T}$.
\end{lemma}
\begin{proof}
    Since $\gf{T}$ is a tree it has no cycles, therefore $\tc_\X$ has no cycles either and we only have to show that it is connected. By construction, any $\widetilde{\gf{T}}$ chosen in Algorithm~\ref{alg:reconstruction} is a clique tree representation of $\lhp$. By the definition of a clique tree, for any $\X$ such that $\bs{\X}\in\pos$, the subgraph of $\widetilde{\gf{T}}$ induced by the set of vertices corresponding to the max-cliques of $\lhp$ that contain $\bs{\X}$ is a subtree of $\widetilde{\gf{T}}$. By definition, $\tcut{\X}$ contains precisely these vertices as well as the boundary vertices corresponding to the parties in $\X$. For any $\ell\in\X$, there is by construction an edge in $\gf{T}$ connecting $v_\ell$ to the vertex $v_Q$ corresponding to the unique max-clique $Q$ of $\lhp$ that contains $Q_\ell$. Since (by definition of $Q^\ell$) $\bs{\X}$ is an element of $Q_\ell$, and $\tcut{\X}$ contains all vertices corresponding to all max-cliques that contain $\bs{\X}$, it also contains $v_Q$. Therefore, the leaf $v_l$ of $\gf{T}$ is connected to a vertex in $\tcut{\X}$ and repeating the same argument for all $\ell\in\X$ we complete the proof. 
\end{proof}

\Cref{lem:positive-bs-connected} pertains to any subsystem whose $\bset$ is positive, and we will compute its cost at the end of this subsection (cf., \Cref{thm:C-equals-S}). 
Consider instead a subsystem $\X$ whose $\bset$ is vanishing. In this case we have 
\begin{equation}
    \tc_\X = \bigoplus_{\ell\in\X} \tc_\ell,
\end{equation}
since $\tcut{\X}$ is simply the collection of all boundary vertices for the parties in $\X$, and in $\gf{T}$ they are all leaves by construction. The decomposition of $\tc_\X$ above then immediately implies that the cost of $\tcut{\X}$ obeys \eqref{eq:cost-non-pos-bsets}. 

On the other hand, if the $\bset$ of $\X$ is partial, note that even if the cut $\tcut{\X}$ is defined as the union of the cuts $\tcut{\Y}$ for the components of $\Gamma(\X)$ (cf., \eqref{eq:cut-partial-bs}), \Cref{lem:positive-bs-connected} does not imply that the connected components of $\tc_\X$ are precisely the subgraphs induced by these cuts. Specifically, at this stage it is not clear whether
\begin{equation}
    \tc_\X = \bigoplus_{\Y\in\Gamma(\X)} \tc_\Y,
\end{equation}
since there might be $\Y_1,\Y_2\in\Gamma(\X)$ such that there is an edge in $\gf{T}$ connecting a vertex in $\tc_{\Y_1}$ to one in $\tc_{\Y_2}$, in which case $\tc_{\Y_1\Y_2}\neq\tc_{\Y_1}\oplus \tc_{\Y_2}$. 

In particular, note that if $\bs{\X}$ is partial, then for $\Svec$ we have $\mi(\Y_1:\Y_2)=0$. But if indeed there is an edge in $\gf{T}$ connecting a vertex in $\tc_{\Y_1}$ to one in $\tc_{\Y_2}$, then the cost of $\tcut{\Y_1\Y_2}$ is strictly less than sum of the costs of $\tcut{\Y_1}$ and $\tcut{\Y_2}$ (since the weight of the edge connecting $\tc_{\Y_1}$ and $\tc_{\Y_2}$ contributes to the costs of $\tc_{\Y_1}$ and $\tc_{\Y_2}$ but not to that of $\tcut{\Y_1\Y_2}$) and the mutual information between $\Y_1$ and $\Y_2$ computed on $\gf{T}$ using the min-cut prescription is necessarily strictly positive, implying that Algorithm~\ref{alg:reconstruction} does not reproduce $\Svec$. 

The next step of the proof will be to show that this situation can never arise. This will be the main content of \Cref{lem:disconneted-cuts} below, but before we prove this result, it is useful to consider a simpler version of this situation, which is the content of \Cref{lem:separation}. To state the lemma, we introduce the following terminology: For a choice of disjoint subsystems $\X,\Y$, we say that $\X$ and $\Y$ are \textit{separated by a subsystem $\Z$ in the tree $\gf{T}$} if $\Z\cap(\X\cup\Y)=\varnothing$, and for any vertex $v_\X$ in $\tcut{\X}$ and any vertex $v_\Y$ in $\tcut{\Y}$, the path in $\gf{T}$ from $v_\X$ to $v_\Y$ contains at least one vertex of $\tcut{\Z}$. 

\begin{lemma}
\label{lem:separation}
    For any disjoint subsystems $\X,\Y$, if $\X$ and $\Y$ are separated by $\Z$ in $\gf{T}$, then $\mi(\X:\Y)(\Svec)=0$.  
\end{lemma}
\begin{proof}
    We proceed by contradiction. Suppose that there are disjoint subsystems $\X,\Y$ separated by a subsystem $\Z$, and that $\mi(\X:\Y)(\Svec)\neq 0$. By the latter assumption, $\bs{\X\Y}$ is either positive of partial. However, it cannot be positive, since otherwise, by \Cref{lem:positive-bs-connected}, $\tc_{\X\Y}$ is connected, and since $\X$ and $\Y$ are separated by $\Z$, it must be that $\tcut{\X\Y}$ intersects $\tcut{\Z}$, contradicting (i) of \Cref{lem:fund-cut-properties}. If instead $\bs{\X\Y}$ is partial, the assumption $\mi(\X:\Y)(\Svec)\neq 0$ implies that there is a subsystem $\X'\Y'$ such that $\bs{\X'\Y'}$ is positive and 
    \begin{equation}
        \varnothing\neq\X'\subseteq\X, \qquad \varnothing\neq\Y'\subseteq\Y.
    \end{equation}
    By \Cref{lem:positive-bs-connected}, $\tc_{\X'\Y'}$ is connected but, by (ii) of \Cref{lem:fund-cut-properties}, we have $\tcut{\X'}\subseteq\tcut{\X}$ and $\tcut{\Y'}\subseteq\tcut{\Y}$, implying that $\X'$ and $\Y'$ are separated by $\Z$, again contradicting (i) of \Cref{lem:fund-cut-properties}.
\end{proof}

Using \Cref{lem:separation} we can then prove the following key intermediate result.

\begin{lemma}
\label{lem:disconneted-cuts}
    For any pair of disjoint subsystems $\X$ and $\Y$
    \begin{equation}
        \tc_{\X\Y}=\tc_{\X}\oplus \tc_{\Y} \quad \iff \quad \mi(\X:\Y)(\Svec)=0.
    \end{equation}
\end{lemma}
\begin{proof}
    ($\Rightarrow$) Consider two disjoint subsystems $\X,\Y$ such that $\tc_{\X\Y}=\tc_{\X}\oplus \tc_{\Y}$. By the contrapositive of \Cref{lem:positive-bs-connected}, $\bs{\X\Y}\notin\pos$. If $\bs{\X\Y}\in\van$, the statement is trivially true. Suppose then that $\bs{\X\Y}\in\parv$ and consider an element $\bs{\Z}$ in the partition $\Gamma(\X\Y)$ such that $|\Z|\geq 2$. By \Cref{lem:positive-bs-connected}, $\gf{T}_\Z$ is a subtree of $\gf{T}$. Since $\Z\subseteq\X\Y$, by (ii) of \Cref{lem:fund-cut-properties} we have $\tcut{\Z}\subseteq\tcut{\X\Y}$, and since $\tc_{\X\Y}$ is disconnected, $\gf{T}_\Z$ must be a subtree of $\tc_{\X}$ or $\tc_{\Y}$. Therefore the bipartition of $\X\Y$ into $\X$ and $\Y$ does not split any element of $\Gamma(\X\Y)$, and $\mi(\X:\Y)(\Svec)=0$.

    ($\Leftarrow$) We proceed by contradiction. Suppose that given $\Svec$ there exists a pair of subsystems $\X,\Y$ such that $\mi(\X:\Y)(\Svec)=0$, and that there is a pair of vertices $v_\X,v_\Y$, with $v_\X\in\tcut{\X}$ and $v_\Y\in\tcut{\Y}$, which are connected by an edge $e$ of $\gf{T}$. If $\bs{\X}$ is partial or vanishing, we denote by $\X_*$ the element of $\Gamma(\X)$ such that $\tcut{\X_*}$ contains $v_\X$; if $\bs{\X}$ is positive, we simply set $\X_*=\X$. The subsystem $\Y_*$ is defined analogously. Notice that $\tc_{\X_*}$ is a subtree of $\gf{T}$, since either $\bs{\X_*}$ is positive (and we can use \Cref{lem:positive-bs-connected}), or $|\X_*|=1$, and $\tc_{\X_*}$ is a single vertex; the same result applies to $\tc_{\Y_*}$.
    
    We now partition the vertex set $V$ of $\gf{T}$ into two parts according to the two connected components of the subgraph obtained from $\gf{T}$ by deleting $e$. Specifically, $V_{\X_*}$ is the set of vertices of the component that contains the vertices in $\tcut{\X_*}$, and $V_{\Y_*}$ is defined analogously (i.e., it is the complement of $V_{\X_*}$ in $V$). Introducing subsystems $\X'$ and $\Y'$ defined by
    \begin{equation}
        \X'=(V_{\X_*}\setminus \tcut{\X_*})\cap\partial V \qquad \Y'=(V_{\Y_*}\setminus \tcut{\Y_*})\cap\partial V,
    \end{equation}
    we obtain a partition $\{\X_*,\X',\Y_*,\Y'\}$ of $\nsp$.\footnote{\,Here we assume that neither $\X'$ nor $\Y'$ is empty, otherwise the argument is simpler. We leave the details as an exercise for the reader.} By construction, $\X'$ is separated from $\Y'$ by $\X_*\Y_*$, and $\X_*$ is separated from $\Y'$ by $\Y_*$.\footnote{\,Several additional subsystems are separated by other subsystems, but considering these separations (and two more below) will suffice.} \Cref{lem:separation} then implies 
    \begin{equation}
    \label{eq:VMIs-from-separation}
        \mi(\X':\Y')(\Svec)=\mi(\X_*:\Y')(\Svec)=0,
    \end{equation}
    which we write as
    \begin{align}
    \label{eq:VMIs-from-separation}
        & \ent_{\X'}\, +\, \ent_{\Y'}\, =\, \ent_{\X'\Y'}\, =\, \ent_{\X_*\Y_*} \nonumber\\
        & \ent_{\X_*}\, + \,\ent_{\Y'}\, =\, \ent_{\X_*\Y'}\, =\, \ent_{\X'\Y*}\;  
    \end{align}
    where we used $\ent_\Z=\ent_{\Z^\complement}$. Again by construction, $\X'$ is separated from $\Y_*$ by $\X_*$, and \Cref{lem:separation} implies that $\mi(\X':\Y_*)(\Svec)=0$. Furthermore, from the assumption $\mi(\X:\Y)(\Svec)=0$, the inclusions $\X_*\subseteq\X$ and $\Y_*\subseteq\Y$, and the fact that the PMI of $\Svec$ is a KC-PMI, it follows that $\mi(\X_*:\Y_*)(\Svec)=0$. Using these relations we then rewrite \eqref{eq:VMIs-from-separation} as
    \begin{align}
        & \ent_{\X'}\, +\, \ent_{\Y'}\, =\, \ent_{\X_*}\, +\, \ent_{\Y_*} \nonumber\\
        & \ent_{\X_*}\, + \,\ent_{\Y'}\, =\, \ent_{\X'}\, +\, \ent_{\Y_*} \; .
    \end{align}
    Subtracting these relations we get $\ent_{\X'} \,=\, \ent_{\X_*}$, and from $\mi(\X':\Y_*\Y')(\Svec)=0$, which follows from \Cref{lem:separation} and the separation of $\X'$ from $\Y_*\Y'$ by $\X_*$, we get
    \begin{equation}
        \ent_{\Y_*\Y'}\,  =\, \ent_{\X'\Y_*\Y'}\, - \, \ent_{\X'}\, \nonumber\\
         =\, \ent_{\X_*}\, -\, \ent_{\X'}\, =\, 0. 
    \end{equation}
    This is in contradiction with the assumption that all components of $\Svec$ (except only for $\ent_{\nsp}$ and $\ent_\varnothing$) are non-vanishing, and completes the proof.
\end{proof}

We can now use \Cref{lem:disconneted-cuts} to show that \eqref{eq:cost-non-pos-bsets} does indeed hold.

\begin{lemma}
\label{lem:cost-decompos}
    For any subsystem $\X$ such that $\bs{\X}$ is either partial or vanishing, the following decomposition holds
    \begin{equation}
    \label{eq:subgraph-decomposition}
        \tc_\X = \bigoplus_{\Y\in\Gamma(\X)} \tc_\Y,
    \end{equation}
    and 
    \begin{equation}
    \label{eq:cost-decompos-eq}
        \norm{\tcut{\X}}=\sum_{\Y\in\Gamma(\X)} \norm{\tcut{\Y}}.
    \end{equation}
\end{lemma}
\begin{proof}
    Consider a subsystem $\X$ such that $\bs{\X}$ is either partial or vanishing, and let $\Y_1$ be an element of $\Gamma(\X)$. By basic properties of $\Gamma$-partitions we have $\mi(\Y_1:\X\setminus\Y_1)(\Svec)=0$, and by \Cref{lem:disconneted-cuts} it follows that $\tc_{\X}=\tc_{\Y_1}\oplus\tc_{\X\setminus\Y_1}$. If $\X\setminus\Y_1$ is a component of $\Gamma(\X)$, the proof is complete. Otherwise, let $\Y_2$ be another component of $\Gamma(\X)$. Since $\Y_2\subset (\X\setminus\Y_1)$, the $\bset$ of $\X\setminus\Y_1$ is not positive, and we have $\mi(\Y_2:\X\setminus(\Y_1\cup\Y_2))(\Svec)=0$. using again \Cref{lem:disconneted-cuts} we obtain, 
    \begin{equation}
        \tc_{\X\setminus\Y_1}=\tc_{\Y_2}\oplus\tc_{\X\setminus(\Y_1\cup\Y_2)} \quad \implies \quad \tc_X = \tc_{\Y_1} \oplus \tc_{\Y_2}\oplus\tc_{\X\setminus(\Y_1\cup\Y_2)}.
    \end{equation}
    Proceeding in this fashion for all elements of $\Gamma(\X)$ we obtain \eqref{eq:subgraph-decomposition}, from which \eqref{eq:cost-decompos-eq} follows immediately.
\end{proof}

Having proved \eqref{eq:cost-non-pos-bsets}, let us now briefly summarize what remains to be shown in this subsection. Since both \eqref{eq:cost-non-pos-bsets} and \eqref{eq:entvec-cost-non-pos-bsets} hold for all subsystems whose $\bset$ is not positive, and since by definition of $\Gamma$-partitions all elements of $\Gamma(\X)$ have positive $\bset$ or are singletons, to prove that ${\sf C}_\X=\ent_\X$ for all subsystems $\X$ it is sufficient to show that this relation holds whenever $\bs{\X}$ is positive or $|\X|=1$. The next result proves that this is indeed the case.

\begin{thm}
\label{thm:C-equals-S}
    For any irreducible entropy vector $\Svec$, and any holographic simple tree graph model $\gf{T}$ resulting from Algorithm~\ref{alg:reconstruction},
    \begin{equation}
    \label{eq:C-equals-S}
        {\sf C}_\X = \ent_\X, \quad \forall\,\X\subseteq\nsp.
    \end{equation}
\end{thm}
\begin{proof}
    The case $\X=\varnothing$ is trivial. For $\X=\nsp$, the assumption that $\Svec$ is irreducible implies $\bs{\nsp}\in\pos$, and $h_{\nsp}$ belongs to all max-cliques of $\lhp$. Therefore $\tcut{\nsp}=V(\gf{T})$, from which we obtain
    \begin{equation}
        {\sf C}_{\nsp}= \norm{V(\gf{T})} = 0 = \ent_{\nsp},
    \end{equation}
    where the last equality follows purity of the full system (cf., \eqref{eq:S-extension2}). 
    
    More generally, recall that we only need to show that \eqref{eq:C-equals-S} holds when $\bs{\X}$ is positive or $|\X|=1$, since in combination with \eqref{eq:cost-non-pos-bsets} and \Cref{lem:cost-decompos}, this implies that \eqref{eq:C-equals-S} holds for every other subsystem. Accordingly, we organize the proof into three cases,\footnote{\,Here we assume $\N\geq 2$, but the case $\N=1$ is trivial.} where (i) and (iii) are defined without loss of generality, since we can swap the definition of $\X$ and $\X^\complement$ as needed.

     \paragraph{(i) $|\X|=1$, $\bs{\X^\complement}\in\pos$:} Let $\X=\ell$ for some party $\ell$. By Algorithm~\ref{alg:reconstruction}, the vertex $v_\ell$ is a leaf of $\gf{T}$ and the weight of the only edge incident with $v_\ell$ is $\ent_\ell$. Therefore, since $\tcut{\ell}=\{v_\ell\}$, it follows immediately that ${\sf C}_\ell=\ent_\ell$.

     For the subsystem $\X^\complement$, since $\bs{\X^\complement}$ is positive and $|\X^\complement|=\N$, it follows that $\X^\complement\cap\Y\neq\varnothing$ for any $\Y$ whose $\bset$ is positive (since $|\Y|\geq 2$). The hyperedge $h_{\X^\complement}$ of $\hp$ then belongs to all max-cliques of $\lhp$, and $\tcut{\X^\complement}$ contains all vertices of $\widetilde{\gf{T}}$. Therefore, $\tcut{\X^\complement}$ contains all vertices $v_\ell'$ for $\ell'\neq\ell$, and ${\sf C}_{\X^\complement}$ is given by the weight of the edge incident with the leaf $v_\ell$, from which ${\sf C}_{\X^\complement}={\sf C}_\ell=\ent_\ell=\ent_{\X^\complement}$.

    \paragraph{(ii) $\min(|\X|,|\X^\complement|)\geq 2$, $\bs{\X}\in\pos$, $\bs{\X^\complement}\in\pos$:} by \Cref{lem:positive-bs-connected}, both $\tc_{\X}$ and $\tc_{\X^\complement}$ are connected. Since $\gf{T}$ is a tree, the set of cut edges of each cut contains a single edge, and in both cases the weight is assigned by Algorithm~\ref{alg:reconstruction} to be $\ent_\X=\ent_{\X^\complement}$.

    \paragraph{(iii) $\min(|\X|,|\X^\complement|)\geq 2$, $\bs{\X}\in\pos$, $\bs{\X^\complement}\notin\pos$:} in this case, we first rewrite ${\sf C}_\X$ as follows
    \begin{equation}
    \label{eq:third-case}
        {\sf C}_\X = \!\sum_{e \in \ce({\sf C}_\X )} \!\!w_e \, = \!\sum_{\Y\in\Upsilon(\X^\complement)} \!\ent_\Y \, = \!\sum_{\substack{\Y\in\Upsilon(\X^\complement) \\ \bs{\Y}\in\pos}} \!\ent_\Y \, + \!\sum_{\substack{\Y\in\Upsilon(\X^\complement)\\ |\Y|=1}} \ent_\Y \, +\, \!\!\sum_{\substack{\Y\in\Upsilon(\X^\complement)\\ |\Y|\geq 2 \\ \bs{\Y}\notin\pos}} \!\ent_\Y\; . 
    \end{equation}
    In the equation above, the second equality follows from the fact that the edges in $\ce({\sf C}_\X)$ with non vanishing weights correspond to pairs $\langle\Y,\Y^\complement\rangle_e$ such that if, for each $e\in\ce({\sf C}_\X)$, we choose $\Y$ such that $\Y\cap\X=\varnothing$, the resulting collection of subsystems is a partition $\Upsilon(\X^\complement)$ of $\X^\complement$. In the third equality, we have simply split the sum into three terms according to the cardinality of $\Y$ and on whether $\bs{\Y}$ is positive or not. Henceforth, to simplify the notation, we will keep implicit the condition $\Y\in\Upsilon(\X^\complement)$ in the remainder of the proof.

    Consider now the partition $\Gamma(\X^\complement)$ and an arbitrary element $\Z$ such that $|\Z|\geq 2$. By the definition of $\Gamma$-partitions, $\bs{\Z}\in\pos$, and by \Cref{lem:positive-bs-connected} $\tc_{\Z}$ is connected. Since $\Z\cap\Y=\varnothing$, \Cref{lem:fund-cut-properties} implies that $\tcut{\X}\cap\tcut{\Z}=\varnothing$, from which it follows that $\Z\subseteq\Y$ for some element $\Y$ of $\Upsilon(\X^\complement)$. If instead $|\Z|=1$, the inclusion $\Z\subseteq\Y$ for some $\Y$ in $\Upsilon(\X^\complement)$ is trivial (since $\Upsilon(\X^\complement)$ is a partition of $\X^\complement$).

    We can then rewrite \eqref{eq:third-case} as 
    \begin{equation}
    \label{eq:third-case-2}
        {\sf C}_\X = \!\sum_{\bs{\Y}\in\pos} \!\ent_\Y \, + \!\sum_{|\Y|=1} \ent_\Y \, + \!\sum_{\substack{|\Y|\geq 2 \\ \bs{\Y}\notin\pos}} \;\sum_{\Z\in\Gamma(\Y)} \!\ent_\Z \;,
    \end{equation}
    where we have used \eqref{eq:entvec-cost-non-pos-bsets}, and from which it follows (again by \eqref{eq:entvec-cost-non-pos-bsets}) that
    \begin{equation}
        {\sf C}_\X = \sum_{\Z\in\Gamma(\X^\complement)} \!\ent_\Z\, = \, \ent_{\X^\complement} \, =\, \ent_\X.
    \end{equation}
\end{proof}

\subsection{The cuts $\tcut{\X}$ are min-cuts on $\gf{T}$}
\label{subsec:min-cut}

In the previous subsection we have shown that for any irreducible entropy vector $\Svec$, and any choice of $\widetilde{\gf{T}}$, the holographic simple tree graph model $\gf{T}$ resulting from Algorithm~\ref{alg:reconstruction} has the property that for each component $\ent_\X$ of $\Svec$, the cost of the cut $\tcut{\X}$ is equal to $\ent_\X$. However, as we anticipated above, this is not yet sufficient to conclude that $\Svec$ is the entropy vector obtained from $\gf{T}$ following the min-cut prescription reviewed in \S\ref{subsec:graphs}, because a priori the cut $\tcut{\X}$ for $\X$ does not have to be a min-cut for $\X$. In fact, as exemplified in \Cref{fig:SA-violation-example}, if we ignore the minimality of the cost, this matching even extends to entropy vectors that violate SA. On the one hand, this matching ultimately relies on the fact that the topology of the tree obtained from Algorithm~\ref{alg:reconstruction} does not depend on the specific choice of the entropy vector, but rather on its PMI. The results of the previous subsection thereby can be interpreted as showing that as long as an entropy vector belongs to a subspace of entropy space corresponding to a chordal KC-PMI, this matching is attained. On the other hand, the fact that for entropy vectors on this subspace that violate SA the prescription does not work ultimately depends on the intimate relationship between SA and min-cut inequalities on simple tree graph models. Indeed, the next result shows that for irreducible entropy vectors, which we have defined such that they obey SA, the costs $\tcut{\X}$ are always min-cuts precisely because of SA. This result completes the proof of \Cref{thm:alg-suff}, and therefore \Cref{thm:main-reduced}.

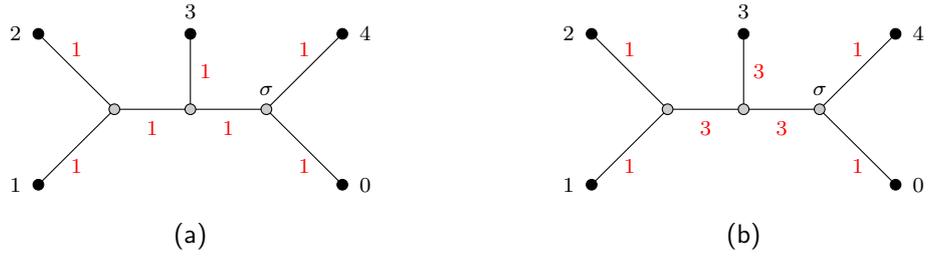
\begin{figure}[tbp]
    \centering
    \begin{subfigure}{0.45\textwidth}
    \centering
    \begin{tikzpicture}[scale=1]
    \draw (0,0) -- (-1,-1);
    \draw (0,0) -- (-1,1);
    \draw (0,0) -- (2,0);
    \draw (2,0) -- (3,1);
    \draw (2,0) -- (3,-1);
    \draw (1,0) -- (1,1);
    
    \filldraw (-1,-1) circle (2pt);
    \filldraw (-1,1) circle (2pt);
    \filldraw (3,-1) circle (2pt);
    \filldraw (3,1) circle (2pt);
    \filldraw (1,1) circle (2pt);
    
    \filldraw[fill=bulkcol] (0,0) circle (2pt);
    \filldraw[fill=bulkcol] (2,0) circle (2pt);
    \filldraw[fill=bulkcol] (1,0) circle (2pt);
    
    \node[] () at (-1.3,1) {{\scriptsize $2$}};
    \node[] () at (-1.3,-1) {{\scriptsize $1$}};
    \node[] () at (1,1.3) {{\scriptsize $3$}};
    \node[] () at (3.3,1) {{\scriptsize $4$}};
    \node[] () at (3.3,-1) {{\scriptsize $0$}};
    
    \node[red] () at (-0.5,0.8) {{\scriptsize $1$}};
    \node[red] () at (-0.5,-0.75) {{\scriptsize $1$}};
    \node[red] () at (0.5,-0.25) {{\scriptsize $1$}};
    \node[red] () at (1.5,-0.25) {{\scriptsize $1$}};
    \node[red] () at (2.5,0.8) {{\scriptsize $1$}};
    \node[red] () at (2.5,-0.75) {{\scriptsize $1$}};
    \node[red] () at (1.2,0.5) {{\scriptsize $1$}};

    \node[] () at (2,0.25) {{\scriptsize $\sigma$}};
    \end{tikzpicture}
    \subcaption[]{}
    \end{subfigure}
    \hspace{0.2cm}
    \begin{subfigure}{0.45\textwidth}
    \centering
    \begin{tikzpicture}[scale=1]
    \draw (0,0) -- (-1,-1);
    \draw (0,0) -- (-1,1);
    \draw (0,0) -- (2,0);
    \draw (2,0) -- (3,1);
    \draw (2,0) -- (3,-1);
    \draw (1,0) -- (1,1);
    
    \filldraw (-1,-1) circle (2pt);
    \filldraw (-1,1) circle (2pt);
    \filldraw (3,-1) circle (2pt);
    \filldraw (3,1) circle (2pt);
    \filldraw (1,1) circle (2pt);
    
    \filldraw[fill=bulkcol] (0,0) circle (2pt);
    \filldraw[fill=bulkcol] (2,0) circle (2pt);
    \filldraw[fill=bulkcol] (1,0) circle (2pt);
    
    \node[] () at (-1.3,1) {{\scriptsize $2$}};
    \node[] () at (-1.3,-1) {{\scriptsize $1$}};
    \node[] () at (1,1.3) {{\scriptsize $3$}};
    \node[] () at (3.3,1) {{\scriptsize $4$}};
    \node[] () at (3.3,-1) {{\scriptsize $0$}};
    
    \node[red] () at (-0.5,0.8) {{\scriptsize $1$}};
    \node[red] () at (-0.5,-0.75) {{\scriptsize $1$}};
    \node[red] () at (0.5,-0.25) {{\scriptsize $3$}};
    \node[red] () at (1.5,-0.25) {{\scriptsize $3$}};
    \node[red] () at (2.5,0.8) {{\scriptsize $1$}};
    \node[red] () at (2.5,-0.75) {{\scriptsize $1$}};
    \node[red] () at (1.2,0.5) {{\scriptsize $3$}};

    \node[] () at (2,0.25) {{\scriptsize $\sigma$}};
    \end{tikzpicture}
    \subcaption[]{}
    \end{subfigure}
    \caption{An example of an entropy vector which violates SA but for which the cost of the cuts $\tcut{\X}$ on the tree $\gf{T}$ resulting from Algorithm~\ref{alg:reconstruction} reproduces the components of the entropy vector. Accordingly, the cuts $\tcut{\X}$ are not min-cuts. For the KC-PMI $\pmi=\;\downarrow\!\{\mi(12:4),\mi(12:0),\mi(04:1),\mi(04:2)\}$  (a) shows the realization of the entropy vector $\Svec=(1, 1, 1, 1, 1, 2, 2, 2, 2, 2, 1, 2, 2, 2, 1)$ obtained from Algorithm~\ref{alg:reconstruction}. (b) is a graph model with the same topology as (a), but with weights assigned from the components of the entropy vector $\vec{\ent}'=(1, 1, 3, 1, 3, 4, 2, 4, 2, 4, 3, 4, 2, 2, 1)$, belonging to the same subspace of entropy space spanned by the face of the SAC corresponding to $\pmi$. Notice that $\vec{\ent}'$ violates SA, since $2=\ent'_1+\ent'_2<\ent'_{12}=3$. Nevertheless, one can easily verify that the cost of the cuts $\tcut{\X}$ correctly reproduce all components of $\vec{\ent}'$. The resolution of the apparent paradox is that while the cut $\tcut{40}=\{0,4,\sigma\}$ is the min-cut for the $40$ subsystem in (a), it is not the min-cut in (b). Accordingly, the min-cut prescription assigns to (b) an entropy vector which is not $\vec{\ent}'$.}
    \label{fig:SA-violation-example}
\end{figure}

\begin{thm}
    For any irreducible entropy vector $\Svec$, any simple tree graph model $\gf{T}$ resulting from Algorithm~\ref{alg:reconstruction}, and any subsystem $\X$, the cut $\tcut{\X}$ is a min-cut for $\X$.
\end{thm}
\begin{proof}
    For $\X=\varnothing$ the statement is trivially true by the non-negativity of the weights. For any $\X\neq\varnothing$, let $C_\X$ be an arbitrary $\X$-cut distinct from $\tcut{\X}$. We need to show that
    \begin{equation}
    \label{eq:min-cut-ineq}
        \norm{\tcut{\X}}\, \leq \norm{C_\X}\qquad \forall\, C_\X.
    \end{equation}
    
    For a choice of subsystem $\X$, consider then an arbitrary cut $C_\X$. The subgraph $\gf{T}_\X$ of $\gf{T}$ induced by $C_\X$ is not necessarily connected, and we decompose $C_\X$ according to the connected components of $\gf{T}_\X$, i.e., 
    \begin{equation}
    \label{eq:cut-decomposition}
        C_\X = \bigcup\, \{C_Y\subseteq C_\X|\, \gf{T}_\Y\; \text{is connected} \},
    \end{equation}
    where for each $C_Y$, the subsystem $\Y$ is the set of parties corresponding to the boundary vertices of $\gf{T}$ in $C_Y$. Notice that $\Y$ can be empty, while by construction, the collection of all non-empty $\Y$ in \eqref{eq:cut-decomposition} is a partition of $\X$.

    For a choice of $C_\X$, suppose now that in the decomposition \eqref{eq:cut-decomposition} there is indeed a term where $\Y$ is the empty set. In this case we obtain a new cut $C'_\X$, with $\norm{C'_\X}\, \leq \norm{C_\X}$ by simply ``absorbing'' the vertices in that term into the complement of the cut $(C_\X)^\complement=V\setminus C_\X$ (since the set of cut edges of $C'_\X$ is a subset of the set of cut edges of $C_\X$, and all weights are non-negative). Furthermore, for any choice of $C_\X$, the decomposition \eqref{eq:cut-decomposition} also applies to $(C_\X)^\complement$, and by a similar reasoning, we can reduce the cost of $C_\X$ by absorbing into $C_\X$ the vertices of any term of the decomposition of $(C_\X)^\complement$ which does not include any boundary vertex. Accordingly, in the remainder of this proof, we will only consider cuts $C_\X$ such that for both $C_\X$ and $(C_\X)^\complement$, the decomposition \eqref{eq:cut-decomposition} does not include terms where $\Y=\varnothing$, since if \eqref{eq:min-cut-ineq} holds under this restriction, it holds in general. To simplify the notation, we keep this restriction implicit, and denote by $C_\X$ an arbitrary cut which obeys these conditions.

    Starting from $\gf{T}$, and for a choice of $C_\X$, we now construct a new vertex-labeled tree $\gf{G}$ as follows. First, we contract all edges of $\gf{T}$ which are not cut by $C_\X$. Next, since each vertex of the resulting tree corresponds to a component of the decomposition \eqref{eq:cut-decomposition}, either for $C_\X$ or for its complement, we label a vertex by the corresponding subsystem $\Y$ in \eqref{eq:cut-decomposition}. The restriction on $C_\X$ discussed above then guarantees that each vertex of $\gf{G}$ is labeled by a non-empty set. Furthermore, each edge $e$ of $\gf{G}$ obviously corresponds to an edge of $\gf{T}$, and the associated pair $\langle\Z,\Z^\complement\rangle_e$ remains unchanged. The tree $\gf{G}$ should not be interpreted as a new holographic graph model; it is merely a device which is convenient for the purpose of this proof.

    The edge set of $\gf{G}$ is precisely the set of cut edges for $C_\X$ in $\gf{T}$, i.e., $E(\gf{G})=\ce(C_\X)$,
    and using \Cref{thm:C-equals-S} and Algorithm~\ref{alg:reconstruction}, we can rewrite \eqref{eq:min-cut-ineq} as
    \begin{equation}
        {\sf C}_\X =\, \ent_\X\, \leq \!\sum_{e\in E(\gf{G})} \ent_{\Z(e)} = \!\sum_{e\in E(\gf{G})} \!w(e)\, =\, \norm{C_\X},
    \end{equation}
    where for an edge $e$ of $\gf{G}$, we denote by $\Z(e)$ an arbitrary element of the pair $\langle\Z,\Z^\complement\rangle_e$. Proving \eqref{eq:min-cut-ineq} then amounts to proving that the components of $\Svec$ satisfy the inequality
    \begin{equation}
    \label{eq:ent-ineq}
        \ent_\X\, \leq \!\sum_{e\in E(\gf{G})} \ent_{\Z(e)}.
    \end{equation}
    In the remainder of the proof we will show that this is the case simply because \eqref{eq:ent-ineq} is a positive linear combination of instances of SA.

    To show this, we will use an edge-labeled rooted tree obtained from $\gf{G}$ as follows. We first choose an arbitrary vertex of $\gf{G}$ corresponding to a component of the complement of $C_\X$, which we designate as the \textit{root} of the tree.  Next, for each edge $e$, we choose the label from $\langle\Z,\Z^\complement\rangle_e$ that does not contain the subsystem labeling the root. We then drop the vertex labels introduced above, since they will no longer play any role. In what follows we denote this edge-labeled rooted tree by $\gf{G}^{(0)}$.
    
    It will be convenient to introduce the following terminology for two particular types of vertices of $\gf{G}^{(0)}$ (as well as of any other graph in the sequence defined below). A vertex which is adjacent to at least two leaves, and such that all its neighbors, except at most for one, are leaves, will be called a \textit{rosette}. If all neighbors of a rosette are leaves, the rosette is said to be \textit{star-like}, otherwise it is \textit{generic}. A vertex with degree 2 which is adjacent to exactly one leaf is a \textit{bud}. Note that the root may or may not be a rosette or a bud.
    
    We now construct a sequence of graphs, where at each step $\gf{G}^{(i)}$ we reduce the number of edges of the graph, until there are no more edges and the sequence ends. Correspondingly, while building this sequence, we explicitly construct the positive linear combination of SA instances that reproduces \eqref{eq:ent-ineq}. Specifically, we construct a sequence of inequalities corresponding to the graphs, where we denote by $\mathfrak{I}^{(i)}$ the inequality at step $i$, and set $\mathfrak{I}^{(0)}$ to be the identity $0\geq 0$. The inequality at the following step, $\mathfrak{I}^{(i+1)}$, is obtained from the previous one by adding either a positive combination of instances of SA, or a single instance (as outlined below). The inequality at the end of the sequence will be \eqref{eq:ent-ineq}, which will complete the proof. 

    The two main transformations that we use to construct the sequences of graphs and inequalities are illustrated in \Cref{fig:min-cut-proof}; a few variations will be discussed below. The first transformation converts a generic rosette into a bud. For a given graph $\gf{G}^{(i)}$ with corresponding inequality $\mathfrak{I}^{(i)}$, and for a chosen rosette, we denote by $\Y$ the label of the edge connecting the rosette to its unique non-leaf neighbor, and by $\X_1,\ldots,\X_k$ (for some $k$ which depends on the graph and choice of rosette) the labels of the edges connecting the rosette to the leaves (cf., \Cref{fig:rosette-in}). The rosette is transformed into a bud by replacing the leaves with a single leaf, with label $\W=\bigcup_{i\in [k]}\X_i$ for the incident edge (cf., \Cref{fig:rosette-out}), resulting in a new graph $\gf{G}^{(i+1)}$. The new inequality $\mathfrak{I}^{(i+1)}$ is obtained from $\mathfrak{I}^{(i)}$ by adding
    \begin{equation}
    \label{eq:rosette-ineq}
        \ent_\W\, \leq\, \sum_{i\in [k]} \ent_{\X_i}\, ,
    \end{equation}
    which is the sum of the following SA instances
    \begin{align}
         \ent_{\X_1\X_2} & \, \leq \, \ent_{\X_1} + \ent_{\X_2} \nonumber\\
         \ent_{\X_1\X_2\X_3} & \, \leq \, \ent_{\X_1\X_2} + \ent_{\X_3} \nonumber\\
         & \, \cdots \nonumber\\
         \ent_{\W} & \, \leq \, \ent_{\X_1\X_2\ldots\X_{k-1}} + \ent_{\X_k}\; .
    \end{align}
    In the particular case of a star-like rosette, the entire star graph is replaced by a single vertex, and the associated inequality is the same as \eqref{eq:rosette-ineq}, but the labels $\X_1,\ldots,\X_k$ now correspond to all neighbors.

\begin{figure}[tbp]
    \centering
    \begin{subfigure}{0.45\textwidth}
    \centering
    \begin{tikzpicture}[scale=1]
    \draw (0.3,0) -- (2,0);
    \draw (2,0) -- (2,1.5);
    \draw (2,0) -- (3.06066,1.06066);
    \draw (2,0) -- (3.5,0);
    \draw (2,0) -- (3.06066,-1.06066);
    \draw (2,0) -- (2,-1.5);
    
    \filldraw[teal] (3.05,-1.05) circle (2pt);
    \filldraw[teal] (3.05,1.05) circle (2pt);
    \filldraw[teal] (2,1.5) circle (2pt);
    \filldraw[teal] (3.5,0) circle (2pt);
    \filldraw[teal] (2,-1.5) circle (2pt);
    \filldraw[orange] (2,0) circle (2pt);
    
    \fill[fill=gray!20] (0.3,0) circle (12pt);
   
    \node[] () at (2.3,1) {{\scriptsize $\X_1$}};
    \node[] () at (2.9,0.5) {{\scriptsize $\X_2$}};
    \node[] () at (3.1,-0.25) {{\scriptsize $\X_3$}};
    \node[] () at (2.5,-0.8) {{\scriptsize $\X_4$}};
    \node[] () at (1.7,-1.1) {{\scriptsize $\X_k$}};

    \filldraw[black] (2.35,-1.5) circle (0.7pt);
    \filldraw[black] (2.59,-1.45) circle (0.7pt);
    \filldraw[black] (2.77,-1.32) circle (0.7pt);

    \node[] () at (1.3,0.25) {{\scriptsize $\Y$}};
    \end{tikzpicture}
    \subcaption[]{}
    \label{fig:rosette-in}
    \end{subfigure}
    \hspace{0.2cm}
    \begin{subfigure}{0.45\textwidth}
    \centering
    \begin{tikzpicture}[scale=1]
    \draw (0.5,0) -- (2,0);
    \draw (2,0) -- (4,0);
    
    \filldraw[teal] (4,0) circle (2pt);
    \filldraw[violet] (2,0) circle (2pt);
    
    \fill[fill=gray!20] (0.5,0) circle (12pt);
   
    \node[] () at (1.4,0.25) {{\scriptsize $\Y$}};
    \node[] () at (3,-0.3) {{\scriptsize $\X_1\cup\ldots\cup \X_k$}};

    \node[] () at (2,-1.5) {};
    \end{tikzpicture}
    \subcaption[]{}
    \label{fig:rosette-out}
    \end{subfigure}
    \\
    \vspace{1cm}
    \begin{subfigure}{0.45\textwidth}
    \centering
     \begin{tikzpicture}[scale=1]
    \draw (0.5,0) -- (2,0);
    \draw (2,0) -- (3.5,0);
    
    \filldraw[teal] (3.5,0) circle (2pt);
    \filldraw[violet] (2,0) circle (2pt);
    
    \fill[fill=gray!20] (0.5,0) circle (12pt);
   
    \node[] () at (1.4,0.25) {{\scriptsize $\Y$}};
    \node[] () at (2.75,0.25) {{\scriptsize $\X$}};

    \end{tikzpicture}
    \subcaption[]{}
    \label{fig:bud-in}
    \end{subfigure}
    \hspace{0.2cm}
    \begin{subfigure}{0.45\textwidth}
    \centering
    \begin{tikzpicture}[scale=1]
    \draw (0.5,0) -- (3,0);
    
    \filldraw[teal] (3,0) circle (2pt);
    
    \fill[fill=gray!20] (0.5,0) circle (12pt);
   
    \node[] () at (2,0.25) {{\scriptsize $\Y\setminus\X$}};

    \end{tikzpicture}
    \subcaption[]{}
    \label{fig:bud-out}
    \end{subfigure}
    \caption{The two main transformations of a graph $\gf{G}^{(i)}$ discussed in the main text. Leaves are shown in green, buds in purple and rosettes in orange. The gray blob represents the rest of the graph, which is assumed here not to be a leaf. In one type of transformation, a generic rosette (a) is transformed into a bud (b) by replacing its adjoining $k$ leaves with a single new leaf, with the edge label shown in the figure (see the main text for the transformation of a star-like rosette). In the second type of transformation, a bud (c) is removed and the remaining edge, connecting the leaf (originally connected to the bud) to the rest of the graph, is labeled as shown in (d).}
    \label{fig:min-cut-proof}
\end{figure}
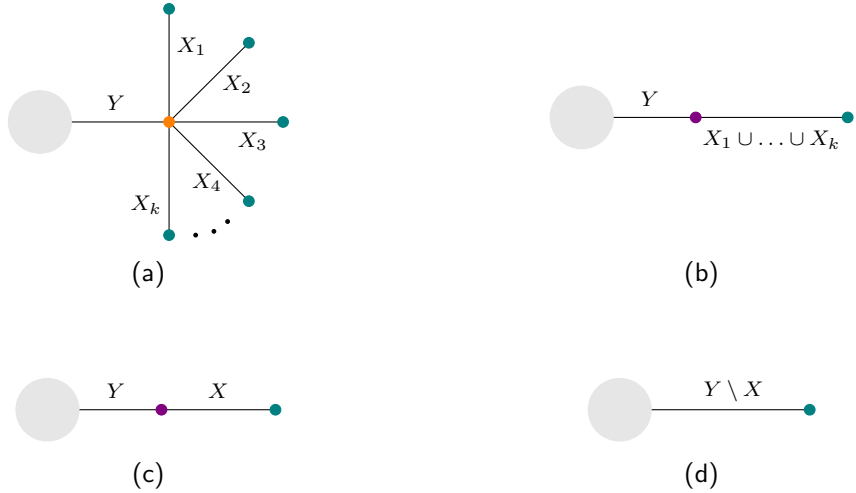

    The second operation removes a bud. For given a graph $\gf{G}^{(i)}$ with corresponding inequality $\mathfrak{I}^{(i)}$, and for a chosen bud, we denote by $\X$ the label of the edge connecting the bud to its leaf, and by $\Y$ the label of the other edge, connecting the bud to the rest of the graph (cf., \Cref{fig:bud-in}). The new graph $\gf{G}^{(i+1)}$ is obtained by replacing the bud and the leaf with a single leaf, with the incident edge labeled by $\Y\setminus\X$ (cf., \Cref{fig:rosette-out}). The new inequality $\mathfrak{I}^{(i+1)}$ is obtained from $\mathfrak{I}^{(i)}$ by adding the following instance of the Araki-Lieb inequality
    \begin{equation}
    \label{eq:bud-ineq}
        \ent_{\Y\setminus\X} \, \leq \, \ent_{Y} + \ent_\X.
    \end{equation}

    Note that in \eqref{eq:rosette-ineq} we have implicitly assumed that for all $i,j\in [k]$ we have $\X_i\cap\X_j=\varnothing$, and similarly, in \eqref{eq:bud-ineq}, that $\Y\supset\X$. However, in general this is not the case, not even for $\gf{G}^{(0)}$.\footnote{\,A simple example is obtained from a vertex-labeled graph $\gf{G}$ which is a star graph, where the central vertex corresponds to the subsystem $\X$ for which we are considering a cut. In this case any possible choice of root is necessarily a leaf, and the label of the edge connecting the root to the center (in the graph $\gf{G}^{(0)}$ obtained from $\gf{G}$) contains all labels of the other edges.} To guarantee that these relations hold at every step of the sequence, we introduce an additional rule governing the choice of rosettes and buds at each step, together with a variant of the main transformations described above for a star-like rosette. Specifically, at each step of the sequence (until the last one), we choose only generic rosettes and buds such that the root is neither the chosen rosette or bud, nor a leaf connected to them (i.e., it is in the gray blob in \Cref{fig:min-cut-proof}). It should be clear that this is always possible until the graph is a star. Furthermore, in the particular case where $\gf{G}^{(i)}$ is a star-like rosette and the root is a leaf (rather than the central vertex), we treat the central vertex as a generic rosette and apply the generic transformation to the non-root leaves, yielding a bud (rather than a single vertex). With these rules, the relations among the labels mentioned above hold for $\gf{G}^{(0)}$ and are preserved at each step of the sequence (we leave the details as a simple exercise for the reader). Note that the complete graph with two vertices $\gf{K}_2$ can be viewed as a star graph, and either of its vertices as a star-like rosette. 

    With these definitions in place, we can now complete the proof. Starting from $\gf{G}^{(0)}$, we iteratively apply the transformations described above until we obtain a star-like rosette with the root in the center,\footnote{\,In the particular case of the graph $\gf{K}_2$ the root in the trivial ``center'' is also a leaf.} and subsequently a single vertex, thereby completing the sequence. It only remains to be shown that, independently of the exact choices of transformations made along the sequence (subject to the constraints mentioned above), the end result is \eqref{eq:ent-ineq}. We show this in three steps: (i) each term on the r.h.s.\ of \eqref{eq:ent-ineq} appears on the r.h.s.\ of a single inequality added along the sequence, and never on the l.h.s.; (ii) all terms appearing in the various inequalities generated along the sequence which do not appear in \eqref{eq:ent-ineq} cancel; (iii) the only term on the l.h.s.\ of the inequality at the end of the sequence matches the l.h.s.\ of \eqref{eq:ent-ineq}. We assume that, for each term on the r.h.s.\ of \eqref{eq:ent-ineq}, the subsystem $\Z(e)$ has been chosen such that it matches the label of the edge $e$ in $\gf{G}^{(0)}$. 

    i) The terms on the r.h.s.\ of \eqref{eq:ent-ineq} correspond to the edge labels of $\gf{G}^{(0)}$. Since each edge is removed, at some step of the sequence, by one of the transformations described above, each corresponding term is added to the final inequality by either \eqref{eq:rosette-ineq} or \eqref{eq:bud-ineq}. Furthermore, since these transformations remove these edges and introduce new ones whose labels are guaranteed to be different from those of $\gf{G}^{(0)}$, each term is added precisely once. Finally, these terms never cancel, because the terms on the l.h.s.\ of \eqref{eq:rosette-ineq} and \eqref{eq:bud-ineq}  correspond only to labels of new edges.

    ii) The transformation that converts a generic rosette into a bud generates a new edge, with a new label, and the corresponding term appears on the l.h.s.\ of \eqref{eq:rosette-ineq}. This transformation however is followed by the removal of the bud, and the inequality \eqref{eq:bud-ineq} includes on the r.h.s.\ the same term. This shows that all terms corresponding to edges generated during the sequence cancel.

    iii) As mentioned above, the sequence ends when the graph $\gf{G}^{(i)}$ is a star graph with the root at the center, and $\gf{G}^{(i+1)}$ is a single vertex. The term $\ent_{\Y}$ on the l.h.s.\ of the last inequality added by the sequence does not cancel, and it is the only term on the l.h.s.\ of the final inequality. Therefore, we only need to show that it matches the l.h.s.\ $\ent_\X$ of \eqref{eq:ent-ineq}, i.e., that $\Y=\X$. We will demonstrate this by showing that every party $\ell\in\X$ belongs to the label of one of the edges of $\gf{G}^{(i)}$, and similarly, that every party $\ell'\notin\X$ does not belong to any of these labels.
    
    First notice that in the vertex-labeled graph $\gf{G}$, and for any pair of adjoining vertices, one vertex corresponds to a component of $C_\X$, and the other to a component of its complement. Therefore, in $\gf{G}^{(0)}$, any vertex such that its distance from the root is odd corresponds to a component of $C_\X$, and any vertex at even distance from the root to a component of the complement. Consider now a party $\ell\in\X$, and let $v$ be the vertex of $\gf{G}$ whose label contains $\ell$. In $\gf{G}^{(0)}$, by construction, the party $\ell$ is contained in the label of every edge in the path from the root to $v$, which has odd length. The transformation of a generic rosette does not change the number of edges whose label contains $\ell$. On the other hand, a first iteration of the transformation of a bud reduces this number by two units (since one edge is deleted and one has a label that no longer contains $\ell$), while a second iteration does not alter it. 
    
    In other to reduce a path of odd length to one of unit length (as required to obtain a leaf of the star $\gf{G}^{(i)}$), this operation needs to be repeated an even number of times. This guarantees that $\ell$ belongs to the label of the edge connecting the root to the leaf of $\gf{G}^{(i)}$ obtained from this reduction. By a similar reasoning, any party $\ell'\notin\X$ is not contained in any such edge, since the transformation of a bud needs to be applied an odd number of times.
\end{proof}

\section{Discussion}
\label{sec:discussion}

We conclude with a few observations about the proof and implications of \Cref{thm:main}, and open questions for future investigations. 

Throughout the proof, and in particular in the formulation of various intermediate lemmas and theorems, we have chosen to always assume that Algorithm~\ref{alg:reconstruction} is applied to irreducible entropy vectors. This choice is motivated, on the one hand, by the clarity it affords to the presentation, and on the other hand by the fact that the main goal of this work is the proof of \Cref{thm:main-reduced}, from which \Cref{thm:main} follows straightforwardly. Nevertheless, it is not clear that all these assumptions are strictly necessary for the proof of each intermediate result presented in \S\ref{sec:proof}, and some may in fact hold more generally (cf.\ the example in \Cref{fig:SA-violation-example}). At present, however, it remains unclear what the application of such generalizations would be, and we therefore leave a systematic investigation of this question to future work.

There is, however, one aspect of our proof related to these possible generalizations that is worth noting. While we have formulated \Cref{thm:main} and \Cref{thm:main-reduced} for entropy vectors that \textit{by assumption} satisfy strong subadditivity, this assumption is not strictly necessary for the proof; indeed, we could have relaxed it to require only that the PMI of $\Svec$ is a KC-PMI (since this is necessary and sufficient for the existence of the correlation hypergraph representation). However, since the entropy vector realized by any holographic graph model \textit{does} obey strong subadditivity, we obtain the following result.

\begin{thm}
    Any entropy vector in the subadditivity cone whose \emph{PMI} is a chordal \emph{KC-PMI} satisfies strong subadditivity.
\end{thm}

Another interesting aspect of the proof of \Cref{thm:main-reduced} is that the specific choice of the clique tree $\widetilde{\gf{T}}$ in Algorithm~\ref{alg:reconstruction} seems to be immaterial. Accordingly, for clarity, we have chosen to keep the dependence on $\widetilde{\gf{T}}$ of various intermediate results implicit throughout the presentation. We believe that the reason why Algorithm~\ref{alg:reconstruction} always constructs a simple tree graph model that reproduces $\Svec$ independently from the choice of $\widetilde{\gf{T}}$ is that, under our conditions, $\widetilde{\gf{T}}$ is in fact unique. However, since the uniqueness of $\widetilde{\gf{T}}$ is not essential for our present purposes, we leave the clarification of this aspect of the algorithm for future work.

Finally, let us comment on the implication of \Cref{thm:main} for the structure of the holographic entropy cone. If the strong form of the conjecture in \cite{Hernandez-Cuenca:2022pst} holds, and therefore every extreme ray of the HEC, for any $\N$, can be realized by a (not necessarily simple) tree graph model, then the result of this work reveals something rather striking. It shows that the “essence’’ of the HEC—its foundational data, the minimal ingredients from which the full structure can be derived, and the guiding principle underlying its geometry—is nothing other than the set of chordal extreme rays of the subadditivity cone. In this sense, the vast and intricate landscape of holographic entropies would ultimately trace its origin back to this remarkably simple and elegant core.

\acknowledgments

V.H. has been supported in part by the U.S. Department of Energy grant DE-SC0009999 and by funds from the University of California.   M.R. acknowledges support from UK Research and Innovation (UKRI) under the UK government’s Horizon Europe guarantee (EP/Y00468X/1).
V.H. acknowledges the hospitality of the Kavli Institute for Theoretical Physics (KITP) during early stages of this work and of the Aspen Center for Physics (supported by National Science Foundation grant PHY-2210452). The authors acknowledge the hospitality of the Centro de Ciencias de Benasque Pedro Pascual during the workshop ``Gravity - New quantum and string perspectives''.
M.R. would like to thank the hospitality of QMAP and the University of California, Davis, during various stages of this work.

There is no underlying data associated with this work.

For the purpose of open access, the authors have applied a Creative Commons Attribution (CC BY) licence to any Author Accepted Manuscript version arising from this submission.

\bibliography{chordality_sufficient}
\bibliographystyle{utphys}

\end{document}